\pdfoutput=1
\documentclass[a4paper,USenglish,cleveref, autoref, thm-restate]{lipics-v2021}

\usepackage{todonotes}

\title{Data Structures Lower Bounds and Popular Conjectures} 

\author{Pavel Dvořák}{Charles University, Prague, Czech Republic}{koblich@iuuk.mff.cuni.cz}{}{}

\author{Michal Koucký}{Charles University, Prague, Czech Republic}{koucky@iuuk.mff.cuni.cz}{}{}

\author{Karel Král}{Charles University, Prague, Czech Republic}{kralka@iuuk.mff.cuni.cz}{}{}

\author{Veronika Slívová}{Charles University, Prague, Czech Republic}{slivova@iuuk.mff.cuni.cz}{}{}

\authorrunning{P. Dvořák, M. Koucký, K. Král, and V. Slívová} 

\Copyright{Pavel Dvořák, Michal Koucký, Karel Král, and Veronika Slívová} 



\relatedversion{} 

\funding{The authors were partially supported by Czech Science Foundation GA{\v C}R grant \#19-27871X. 
This project has received funding from the European Union’s Horizon 2020 research and innovation programme under the Marie Skłodowska-Curie grant agreement No. 823748.}

\acknowledgements{We would like to thank to Mike Saks and Sagnik Mukhopadhyay for insightful discussions.}

\nolinenumbers 

\hideLIPIcs  

\EventEditors{John Q. Open and Joan R. Access}
\EventNoEds{2}
\EventLongTitle{42nd Conference on Very Important Topics (CVIT 2016)}
\EventShortTitle{CVIT 2016}
\EventAcronym{CVIT}
\EventYear{2016}
\EventDate{December 24--27, 2016}
\EventLocation{Little Whinging, United Kingdom}
\EventLogo{}
\SeriesVolume{42}
\ArticleNo{23}

\newcommand{\DS}{{\cal D}}
\newcommand{\Adv}{\mathbf{a}}
\newcommand{\Advf}[1]{\Adv_#1}
\newcommand{\Sp}{s}
\newcommand{\Tm}{t}
\newcommand{\R}{\mathbb{R}}
\newcommand{\E}{\mathbb{E}}
\newcommand{\F}{\mathbb{F}}
\newcommand{\N}{\mathbb{N}}
\newcommand{\un}{\mathsf{un}}
\newcommand{\cF}{{\cal F}}
\newcommand{\cI}{{\cal I}}
\newcommand{\cT}{{\cal T}}
\newcommand{\PI}{\mathsf{Inv}}
\newcommand{\PE}{\mathsf{Eval}}
\newcommand{\PInt}{\mathsf{Interp}}
\newcommand{\cp}{\mathsf{cell}}

\newtheorem{fact}{Fact}

\newcommand{\problem}[5]{

\smallskip
\begin{tabular}{|l p{8cm}|}
	\hline
	\multicolumn{2}{|l|}{\textbf{#1}} \\
	\textit{Input:\ }&{#2}\\
	\textit{Preprocessing:\ } & {#3} \\
	\textit{Query:\ }&{#4}\\
	\textit{Answer:\ }&{#5}\\
	\hline
\end{tabular}
\smallskip

}

\begin{document}

\maketitle

\begin{abstract}
In this paper, we investigate the relative power of several conjectures that attracted recently lot of interest. We establish a connection between the Network Coding Conjecture ({\em NCC}) of Li and Li~\cite{Li10} and several data structure like problems such as non-adaptive function inversion of Hellman~\cite{Hellman80} and the well-studied problem of polynomial evaluation and interpolation. In turn these data structure problems imply super-linear circuit lower bounds for explicit functions such as integer sorting and multi-point polynomial evaluation.
\end{abstract}

\section{Introduction}
One of the central problems in theoretical computer science is proving lower bounds in various models of computation such as circuits and data structures. Proving super-linear size lower bounds for circuits even when their depth is restricted is rather elusive. Similarly, proving polynomial lower bounds on query time for certain static data structure problems seems out of reach. To deal with this situation researchers developed various conjectures which if true would imply the sought after lower bounds. In this paper, we investigate the relative power of some of those conjectures. We establish a connection between the Network Coding Conjecture ({\em NCC}) of Li and Li~\cite{Li10} used recently to prove various lower bounds such as lower bounds on circuit size counting multiplication~\cite{Afshani19} and a number of IO operations for external memory sorting \cite{Farhadi19}.

Another problem researchers looked at is a certain data structure type problem for function inversion~\cite{Hellman80} which is popular in cryptography.
Corrigan-Gibbs and Kogan \cite{Gibbs19} observed that lower bounds for the function inversion problem imply lower bounds for logarithmic depth circuits.
In this paper we establish new connections between the problems, and identify some interesting instances.
Building on the work of Afshani et al.~\cite{Afshani19} we show that the Network Coding Conjecture implies certain weak lower bounds for the inversion data structure problems. That in turn implies the same type of circuit lower bounds as given by Corrigan-Gibbs and Kogan~\cite{Gibbs19}.
We show that similar results apply to a host of other data structure problems such as the well-studied polynomial evaluation problem or the Finite Field Fourier transform problem. 
Corrigan-Gibbs and Kogan~\cite{Gibbs19} gave their circuit lower bound for certain apriori undetermined function. We establish the same circuit lower bounds for sorting integers which is a very explicit function. Similarly, we establish a connection between data structure for polynomial evaluation and circuits for multi-point polynomial evaluation. Our results sharpen and generalize the picture emerging in the literature.

The data structure problems we consider in this paper are for static, non-adaptive, systematic data structure problems, a very restricted class of data structures for which lower bounds should perhaps be easier to obtain. Data structure problems we consider have the following structure: Given the input data described by $N$ bits, create a data structure of size $s$. Then we receive a single query from a set of permissible queries and we are supposed to answer the query while non-adaptively inspecting at most $t$ locations in the data structure and in the original data. The non-adaptivity means that the inspected locations are chosen only based on the query being answered but not on the content of the inspected memory.
We show that when $s \ge \omega\bigl(N \log N/\log \log N\bigr)$, polynomial lower bounds on $t$ for certain problems would imply super-linear lower bounds on 
log-depth circuits for computing sorting, multi-point polynomial evaluation, and other problems. 

We show that logarithmic lower bounds on $t$ for the data structures can be derived from the Network Coding Conjecture even in the more generous setting of $s \ge \varepsilon N$ and when inspecting locations in the data structure is for free. This matches the lower bounds of Afshani~\cite{Afshani19} for certain circuit parameters derived from the Network Coding Conjecture.
One can recover the same type of result they showed from our connection between the Network Coding Conjecture, data structure lower bounds, and circuit lower bounds. 

In this regard, the Network Coding Conjecture seems the strongest among the conjectures, which is the hardest to prove. One would hope that for the strongly restricted data structure problems, obtaining the required lower bounds should be within our reach.

\smallskip\noindent {\bf Organization.} This paper is organized as follows. In the next section we review the data structure problems we consider. Then we provide a precise definition of Network Coding Conjecture in Section~\ref{sec:NCC}. Section~\ref{sec:results} contains the statement of our main results. In Sections~\ref{sec:NCC_function_inversion} and~\ref{sec:NCC_polynomial} we prove our main result for the function inversion and the polynomial problems. In Section~\ref{sec:valiant} we discus the connection between data structure and circuit lower bounds for explicit functions. 

\section{Data Structure Problems}

In this paper, we study lower bounds on \emph{systematic data structures} for various problems -- function inversion, polynomial evaluation, and polynomial interpolation.
We are given an input $I = \{x_0,\dots,x_{n-1}\}$, where each $x_i \in [n] = \{0,\dots,n-1\}$ or each $x_i$ is an element of some field $\F$.
First, a data structure algorithm can preprocess $I$ to produce an advice string $\Advf{I}$ of $\Sp$ bits (we refer to the parameter $\Sp$ as \emph{space} of the data structure $\DS$).
Then, we are given a query $q$ and the data structure should produce a correct answer (what is a correct answer depends on the problem).
To answer a query $q$, the data structure $\DS$ has access to the whole advice string $\Advf{I}$ and can make $\Tm$ queries to the input $I$, i.e., read at most $\Tm$ elements from $I$.
We refer to the parameter $\Tm$ as query time of the data structure.

We consider non-uniform data structures as we want to provide connections between data structures and non-uniform circuits.
Formally, a non-uniform systematic data structure $\DS_n$ for an input $I = \{x_0,\dots,x_{n-1}\}$ is a pair of algorithms $({\cal P}_n, {\cal Q}_n)$ with oracle access to $I$.
The algorithm ${\cal P}_n$ produces the advice string $\Advf{I} \in \{0,1\}^\Sp$.
The algorithm ${\cal Q}_n$ with inputs $\Advf{I}$ and a query $q$ outputs a correct answer to the query $q$ with at most $\Tm$ oracle queries to $I$.
The algorithms ${\cal P}_n$ and ${\cal Q}_n$ can differ for each $n \in \N$.

\subsection{Function Inversion}
In the function inversion problem, we are given a function $f: [n] \to [n]$ and a point $y \in [n]$ and we want to find $x \in [n]$ such that $f(x) = y$.
This is a central problem in cryptography as many cryptographic primitives rely on the existence of a function that is hard to invert.
To sum up we are interested in the following problem.
\problem{Function Inversion}
{A function $f: [n] \to [n]$ as an oracle.}
{Using $f$, prepare an advice string $\Advf{f} \in \{0,1\}^\Sp$.}
{Point $y \in [n]$.}
{Compute the value $f^{-1}(y)$, with a full access to $\Advf{f}$ and using at most $\Tm$ queries to the oracle for $f$.}

We want to design an efficient data structure, i.e., make $\Sp$ and $\Tm$ as small as possible.
There are two trivial solutions.
The first one is that the whole function $f^{-1}$ is stored in the advice string $\Advf{f}$, thus $\Sp = O(n \log n)$ and $\Tm = 0$.
The second one is that the whole function $f$ is queried during answering a query $y \in [n]$, thus $\Tm = O(n)$ and $\Sp = 0$.
Note that the space $\Sp$ of the data structure is the length of the advice string $\Advf{f}$ in bits, but with one oracle-query $x_i$ the data structure reads the whole $f(x_i)$, thus with $n$ oracle-queries we read the whole description of $f$, i.e., $n \log n$ bits. 

The question is whether we can design a data structure with $\Sp,\Tm \leq o(n)$.
Hellman~\cite{Hellman80} gave the first non-trivial solution and introduced a randomized systematic data structure which inverts a function with a constant probability (over the uniform choice of the function $f$ and the query $y \in [n]$) and $\Sp = O\left(n^{2/3} \log n\right)$ and $\Tm = O \left(n^{2/3} \log n \right)$.
Fiat and Naor~\cite{Fiat99} improved the result and introduced a data structure that inverts any function at any point, however with a slightly worse trade-off: $\Sp^3 \Tm = O\left(n^3 \log n\right)$.
Hellman~\cite{Hellman80} also introduced a more efficient data structure for inverting a permutation -- it inverts any permutation at any point and $\Sp\Tm = O(n\log n)$. 
Thus, it seems that inverting a permutation is an easier problem than inverting an arbitrary function.

In this paper, we are interested in lower bounds for the inversion problem.
Yao~\cite{Yao90} gave a lower bound that any systematic data structure for the inversion problem must have $\Sp\Tm \geq \Omega(n\log n)$, however, the lower bound is applicable only if $\Tm \leq O(\sqrt{n})$.
Since then, only slight progress was made.
De et al.~\cite{De10} improved the lower bound of Yao~\cite{Yao90} that it is applicable for the full range of $\Tm$.
Abusalah et al.~\cite{Abusalah17} improved the trade-off, that for any $k$ it must hold that $\Sp^k\Tm \geq \Omega\left(n^k\right)$.
Seemingly, their result contradicts Hellman's trade-off $\left(\Sp = \Tm = O\left(n^{2/3} \log n\right)\right)$ as it implies $s = t \geq n^{k/k+1}$ for any $k$.
However, for Hellman's attack~\cite{Hellman80} we need that the function can be efficiently evaluated and the functions introduced by Abusalah et al.~\cite{Abusalah17} cannot be efficiently evaluated.
There is also a series of papers~\cite{Gennaro05,Unruh07,Dodis17,Coretti18} which study how the probability of successful inversion depends on the parameters $\Sp$ and $\Tm$.
However, none of these results yields a better lower bound than $\Sp\Tm \geq \Omega(n \log n)$.
Hellman's trade-off is still the best known upper bound trade-off for the inversion problem.
Thus, there is still a substantial gap between the lower and upper bounds.

Another caveat of all known data structures for the inversion is that they heavily use adaptivity during answering queries $y \in [n]$.
I.e., queries to the oracle depend on the advice string $\Adv$ and answers to the oracle queries which have been already made.
We are interested in non-adaptive data structures.
We say a systematic data structure is \emph{non-adaptive} if all oracle queries depend only on the query $y \in [n]$.

As non-adaptive data structures are weaker than adaptive ones, there is a hope that for non-adaptive data structures we could prove stronger lower bounds.
Moreover, the non-adaptive data structure corresponds to circuits computation~\cite{valiant1977graph,valiant1983exponential,Viola09,Gibbs19}.
Thus, we can derive a circuit lower bound from a strong lower bound for a non-adaptive data structure.
Non-adaptive data structures were considered by Corrigan-Gibbs and Kogan~\cite{Gibbs19}.
They proved that improvement by a polynomial factor of Yao's lower bound~\cite{Yao90} for non-adaptive data structures would imply the existence of a function $F: \{0,1\}^N \to \{0,1\}^N$ for $N = n \log n$  that cannot be computed by a linear-size and logarithmic-depth circuit.
More formally, they prove that if a function $f: [n] \to [n]$ cannot be inverted by a non-adaptive data structure of space $O\left(n \log n / \log \log n\right)$ and query time $O(n^\varepsilon)$ for some $\varepsilon > 0$ then there exists a function $F: \{0,1\}^N \to \{0,1\}^N$  that cannot be computed by any circuit of size $O(N)$ and depth $O(\log N)$.
They interpret $r \in \{0,1\}^N$ as $n$ numbers in $[n]$, i.e, $r = (r_1,\dots, r_n) \in \{0,1\}^N$ where each $r_i \in [n]$.
The function $F$ is defined as $F(y) = F(y_1,\dots,y_n) = \bigl(f^{-1}(y_1),\dots.f^{-1}(y_n)\bigr)$ where $f^{-1}(y_i) = \min \bigl\{ x \in [n] \mid f(x) = y \bigr\}$ and $\min \emptyset = 0$.
Informally, if the function $f$ is hard to invert at some points, then it is hard to invert at all points together.
Moreover, they showed equivalence between function inversion and substring search.
A data structure for the function inversion of space $\Sp$ and query time $\Tm$ yields a data structure of space $O(\Sp \log \Sp)$ and query time $O(\Tm \log \Tm)$ for finding pattern of length $O(\log n)$ in a binary text of length $O(n \log n)$ and vice versa -- an efficient data structure for the substring search would yield an efficient data structure for the function inversion.
Compared to results of Corrigan-Gibbs and Kogan~\cite{Gibbs19}, we provide an explicit function (sorting integers) which will require large circuits if any of the functions $f$ is hard to invert.

Another connection between data structures and circuits was made by Viola~\cite{Viola19} who considered constant depth circuits with arbitrary gates.

\subsection{Evaluation and Interpolation of Polynomials}
\label{sec:Polynomials}
	In this section, we describe two natural problems connected to polynomials.
	We consider our problems over a finite field $\F$ to avoid issues with encoding reals.

	\problem{Polynomial Evaluation over $\F$}
	{Coefficients of a polynomial $p \in \F[x]$: $\alpha_0, \ldots, \alpha_{n-1} \in \F$ (i.e., $p(x) = \sum_{i \in [n]} \alpha_i x^{i}$)}
	{Using the input, prepare an advice string $\Advf{p} \in \{0,1\}^\Sp$.}
	{A number $x\in \F$.}
        {Compute the value $p(x)$, with a full access to $\Advf{p}$ and using at most $\Tm$ queries to the coefficients of $p$.}

	\problem{Polynomial Interpolation over $\F$}
	{Point-value pairs of a polynomial $p \in \F[x]$ of degree at most $n-1$: $\bigl(x_0, p(x_0)\bigr), \ldots, \bigl(x_{n-1}, p(x_{n-1})\bigr) \in \F \times \F$ where $x_i \neq x_j$ for any two indices $i \neq j$}
	{Using the input, prepare an advice string $\Advf{p} \in \{0,1\}^\Sp$.}
	{An index $j \in [n]$.}
        {Compute $j$-th coefficient of the polynomial $p$, i.e., the coefficient of $x^{j}$ in $p$, with a full access to $\Advf{p}$ and using at most $\Tm$ queries to the oracle for point-value pairs.}
				In the paper we often use a version of polynomial interpolation where the points $x_0, x_1, \ldots, x_{n-1}$ are fixed in advance and the input consists just of $p(x_0), p(x_1), \ldots, p(x_{n-1})$.
				Since we are interested in lower bounds, this makes our results slightly stronger.

	Let $\F = \textsc{GF}(p^k)$ denote the Galois Field of $p^k$ elements.
	Let $n$ be a divisor of $p^k - 1$.
	It is a well-known fact that for any finite field $\F$ its multiplicative group $\F^*$ is cyclic (see e.g. Serre~\cite{serre2012course}).
	Thus, there is an element $\sigma \in \F$ of order $n$ in the multiplicative group $\F^*$ (that is an element $\sigma$ such that $\sigma^n = 1$ and for each $1 \leq j < n$, $\sigma^j \neq 1$).
	In other words, $\sigma$ is our choice of primitive $n$-th root of unity.
	Pollard~\cite{pollard1971fast} defines the \emph{Finite Field Fourier transform} (FFFT) (with respect to $\sigma$) as a linear function $\text{FFFT}_{n,\sigma} \colon \F^n \rightarrow \F^n$ which satisfies:
	\begin{align*}
		\text{FFFT}_{n,\sigma}(\alpha_0, \ldots, \alpha_{n-1}) &= (\beta_0, \ldots, \beta_{n-1}) \text{ where } \\
		\beta_i &= \sum_{j \in [n]} \alpha_j \sigma^{ij} & \text{for any $i \in [n]$}
	\end{align*}
	The inversion $\text{FFFT}_{n,\sigma}^{-1}$ is given by:
	\begin{align*}
		\text{FFFT}_{n,\sigma}^{-1}(\beta_0, \ldots, \beta_{n-1}) &= (\alpha_0, \ldots, \alpha_{n-1}) \text{ where } \\
		\alpha_i &= \frac{1}{n} \sum_{j \in [n]} \beta_{j} \sigma^{-ij} & \text{for any $i \in [n]$}
	\end{align*}
Note that if we work over a finite field $\F$, our $n$ might not be an element of $\F$.
For simplicity we slightly abuse the notation and use $\frac{1}{n} = \left( \sum_{i=1}^n 1 \right)^{-1}$.  
In our theorems we always set $n$ to be a divisor of $|\F| - 1 = p^k - 1$ thus $n$ modulo $p$ is non-zero and the inverse exists.
	Observe, that $\text{FFFT}_{n,\sigma}^{-1} = \frac{1}{n} \text{FFFT}_{n, \sigma^{-1}}$.
        Hence, FFFT is the finite field analog of Discrete Fourier transform (DFT) which works over complex numbers.

	The FFT algorithm by Cooley and Tukey~\cite{cooley1965algorithm} can be used for the case of finite fields as well (as observed by Pollard~\cite{pollard1971fast}) to get an algorithm using $O(n \log n)$ field operations (addition or multiplication of two numbers).
	Thus we can compute $\text{FFFT}_{n, \sigma}$ and its inverse in $O(n \log n)$ field operations.

	It is easy to see that $\text{FFFT}_{n, \sigma}$ is actually evaluation of a polynomial in multiple special points (specifically in $\sigma^0, \ldots, \sigma^{n-1}$).
	We can also see that it is a special case of interpolation by a polynomial in multiple special points since $\text{FFFT}_{n,\sigma}^{-1} = \frac{1}{n} \text{FFFT}_{n, \sigma^{-1}}$.
	We provide an NCC-based lower bound for data structures computing the polynomial evaluation. 
	However, we use the data structure only for evaluating a polynomial in powers of a primitive root of unity. 
	Thus, the same proof yields a lower bound for data structures computing the polynomial interpolation.


	There is a great interest in data structures for polynomial evaluation in a cell probe model.
In this model, some representation of a polynomial $p = \sum_{i \in [n]} \alpha_i x^i \in \F[x]$ is stored in a table $\cT$ of $\Sp_\cp$ cells, each of $w$ bits.
Usually, $w$ is set to $O\bigl(\log |\F|\bigr)$, that we can store an element of $\F$ in a single cell.
On a query $x \in \F$ the data structure should output $p(x)$ making at most $\Tm_\cp$ probes to the table $\cT$.
A difference between data structures in the cell probe model and systematic data structures is that a data structure in the cell probe model is charged for any probe to the table $\cT$ but a systematic data structure is charged only for queries to the input (the coefficients $\alpha_i$), reading from the advice string $\Advf{p}$ is for free.
Note that, the coefficients $\alpha_i$ of $p$ do not have to be even stored in the table $\cT$.
There are again two trivial solutions.
The first one is that we store a value $p(x)$ for each $x \in \F$ and on a query $x \in \F$ we probe just one cell.
Thus, we would get $\Tm_\cp = 1$ and $\Sp_\cp = |\F|$ (we assume that we can store an element of $\F$ in a single cell).
The second one is that we store the coefficients of $p$ and on a query $x \in \F$ we probe all cells and compute the value $p(x)$.
Thus, we would get $\Tm_\cp = \Sp_\cp = n$.

Let $k = \log |\F|$.
Kedlaya and Umans~\cite{Kedlaya08} provided a data structure for the polynomial evaluation that uses space $n^{1 + \varepsilon}\cdot k^{1 + o(1)}$ and query time $\log^{O(1)} n \cdot k^{1 + o(1)}$.
Note that, $n \cdot k$ is the size of the input and $k$ is the size of the output.

The first lower bound for the cell probe model was given by Miltersen~\cite{Miltersen95}.
He proved that for any cell probe data structure for the polynomial evaluation it must hold that $\Tm_\cp \geq \Omega\bigl(k/\log \Sp_\cp\bigr)$.
This was improved by Larsen~\cite{Larsen12} to $\Tm_\cp \geq \Omega\bigl(k/\log(\Sp_\cp w/nk)\bigr)$, that gives $\Tm_\cp \geq \Omega(k)$ if the data structure uses linear space $\Sp_\cp\cdot w = O(n\cdot k)$. However, the size of $\F$ has to be super-linear, i.e., $|\F| \geq n^{1 + \Omega(1)}$.
Data structures in a bit probe model were studied by G{\'a}l and Miltersen~\cite{Gal03}.
The bit probe model is the same as the cell probe model but each cell contains only a single bit, i.e., $w = 1$.
They studied succinct data structures that are data structures such that $\Sp_\cp = (n + r)\cdot k$ for $r < o(n)$.
Thus, the succinct data structures are related to systematic data structures but still, the succinct data structures are charged for any probe (as any other data structure in the cell probe model).
Note that a succinct data structure stores only a few more bits than it is needed due to information-theoretic requirement.
G{\'a}l and Miltersen~\cite{Gal03} showed that for any succinct data structure in the bit probe model it holds that $r \cdot \Tm_\cp \geq \Omega(n \cdot k)$.
We are not aware of any lower bound for systematic data structures for the polynomial evaluation.

Larsen et al.~\cite{Larsen18} also gives a log-squared lower bound for dynamic data structures in the cell probe model. Dynamic data structures also support updates of the polynomial $p$.

There is a great interest in algorithmic questions about the polynomial interpolation such as how fast we can interpolate polynomials~\cite{Gathen13,BenOr88,Grigoryev90}, how many queries we need to interpolate a polynomial if it is given by oracle~\cite{Clausen91,Ivanyos18}, how to compute the interpolation in a numerically stable way over infinite fields~\cite{Smoktunowicz07} and many others.
However, we are not aware of any results about data structures for the interpolation, i.e., when the interpolation algorithm has an access to some precomputed advice.

\section{Network Coding}
\label{sec:NCC}

We prove our conditional lower bounds based on the Network Coding Conjecture. 
In network coding, we are interested in how much information we can send through a given network.
A \emph{network} consists of a graph $G = (V,E)$, positive capacities of edges $c: E \to \R^+$ and $k$ pairs of vertices $(s_0,t_0),\dots,(s_{k-1},t_{k-1})$.
We say a network $R = \bigl(G,c,(s_i,t_i)_{i \in [k]}\bigr)$ is \emph{undirected} or \emph{directed (acyclic)} if the graph $G$ is undirected or directed (acyclic).
We say a network is $\emph{uniform}$ if the capacities of all edges in the network equal to some $q \in \R^+$ and we denote such network as $\bigl(G,q,(s_i,t_i)_{i \in [k]}\bigr)$.

A goal of a coding scheme for directed acyclic network $R = \bigl(G,c,(s_i,t_i)_{i \in [k]}\bigr)$ is that at each target $t_i$ it will be possible to reconstruct an input message $w_i$ which was generated at the source $s_i$.
The coding scheme specifies messages sent from each vertex along the outgoing edges as a function of received messages.
Moreover, the length of the messages sent along the edges have to respect the edge capacities.

More formally, each source $s_i$ of a network receives an input message $w_i$ sampled (independently of the messages for the other sources) from the uniform distribution $\mathbf{W}_i$ on a set $W_i$.
Without loss of generality we can assume that each source $s_i$ has an in-degree 0 (otherwise we can add a vertex $s'_i$ and an edge $(s'_i, s_i)$ and replace $s_i$ by $s'_i$).
There is an alphabet $\Sigma_e$ for each edge $e \in E(G)$.
For each source $s_i$ and each outgoing edge $e = (s_i,u)$ there is a function $f_{s_i,e}: W_i \to \Sigma_e$ which specifies the message sent along the edge $e$ as a function of the received input message $w_i \in W_i$.
For each non-source vertex $v \in V, v \neq s_i$ and each outgoing edge $e = (v,u)$ there is a similar function $f_{v,e}: \prod_{e' = (u',v)} \Sigma_{e'} \to \Sigma_e$ which specifies the message sent along the edge $e$ as a function of the messages sent to $v$ along the edges incoming to $v$.
Finally, each target $t_i$ has a decoding function $d_i: \prod_{e' = (u',t_i)} \Sigma_{e'} \to W_i$.
The coding scheme is executed as follows:
\begin{enumerate}
 \item Each source $s_i$ receives an input message $w_i \in W_i$. Along each edge $e = (s_i,u)$ a message $f_{s_i,e}(w_i)$ is sent.
 \item When a vertex $v$ receives all messages $m_1,\dots,m_a$ along all incoming edges $(u',v)$ it sends along each outgoing edge $e = (v,u)$ a message $f_{v,e}(m_1,\dots,m_a)$.
 As the graph $G$ is acyclic, this procedure is well-defined and each vertex of non-zero out-degree will eventually send its messages along its outgoing edges.
 \item At the end, each target $t_i$ computes a string $\tilde{w}_i = d_i(m'_1,\dots,m'_b)$ where $m'_j$ denotes the received messages along the incoming edges $(u',t_i)$. We say the encoding scheme is \emph{correct} if $\tilde{w}_i = w_i$ for all $i \in [k]$ and any input messages $w_0,\dots,w_{k-1} \in W_0 \times \cdots \times W_{k-1}$.
\end{enumerate}
The coding scheme has to respect the edge capacities, i.e., if $\mathbf{M}_e$ is a random variable that represents a message sent along the edge $e$, then $H(\mathbf{M}_e) \leq c(e)$, where $H(\cdot)$ denotes the Shannon entropy.
A \emph{coding rate} of a network $R$ is the maximum $r$ such that there is a correct coding scheme for input random variables $\mathbf{W}_0,\dots,\mathbf{W}_{k-1}$ where  $H(\mathbf{W}_i) = \log |W_i| \geq r$ for all $i \in [k]$.
A network coding can be defined also for directed cyclic networks or undirected networks but we will not use it here.

Network coding is related to multicommodity flows.
A multicommodity flow for an undirected network $\bar{R} = \bigl(\bar{G}, c, (s_i, t_i)_{i \in [k]}\bigr)$ specifies flows for each commodity $i$ such that they transport as many units of commodity from $s_i$ to $t_i$ as possible.
A flow of the commodity $i$ is specified by a function $f^i: V \times V \to \R_0^+$ which describes for each pair of vertices $(u,v)$ how many units of the commodity $i$ are sent from $u$ to $v$.
Each function $f^i$ has to satisfy:
\begin{enumerate}
 \item If $u,v$ are not connected by an edge, then $f^i(u,v) = f^i(v,u) = 0$.
 \item For each edge $\{u,v\} \in E(\bar{G})$, it holds that $f^i(u,v) = 0$ or $f^i(v,u) = 0$.
 \item For each vertex $v$ that is not the source $s_i$ or the target $t_i$, it holds that what comes to the vertex $v$ goes out from the vertex $v$, i.e.,
 \[
\sum_{u \in V} f^i(u,v) = \sum_{u \in V} f^i(v,u).  
 \]
 \item What is sent from the source $s_i$ arrives to the target $t_i$, i.e., 
 \[
\sum_{u \in V} f^i(s_i,u) - f^i(u,s_i) = \sum_{u \in V} f^i(u,t_i) - f^i(t_i,u).  
 \]
\end{enumerate}
Moreover, all flows together have to respect the capacities, i.e., for each edge $e = \{u,v\} \in E(\bar{G})$ it must hold that $\sum_{i \in [k]} f^i(u,v) + f^i(v,u) \leq c(e)$.
A \emph{flow rate} of a network $\bar{R}$ is the maximum $r$ such that there is a multicommodity flow $F = (f^0,\dots,f^{k-1})$ that for each $i$ transports at least $r$ units of the commodity $i$ from $s_i$ to $t_i$, i.e., for all $i$, it holds that $\sum_{u \in V} f^i(u,t_i) - f^i(t_i,u) \geq r$.
A multicommodity flow for directed graphs is defined similarly, however, the flows can transport the commodities only in the direction of edges.

Let $R$ be a directed acyclic network of a flow rate $r'$.
It is clear that for a coding rate $r$ of $R$ it holds that $r \geq r'$.
As we can send the messages without coding and thus reduce the encoding problem to the flow problem.
The opposite inequality does not hold: There is a directed network $R = \bigl(G,c,(s_i,t_i)_{i \in [k]}\bigr)$ such that its coding rate is $\Omega\bigl(|V(G)|\bigr)$-times larger than its flow rate as shown by Adler et al.~\cite{Adler06}.
Thus, the network coding for directed networks provides an advantage over the simple solution given by the maximum flow.
However, such a result is not known for undirected networks.
Li and Li~\cite{Li10} conjectured that the network coding does not provide any advantage for undirected networks, thus for any undirected network $\bar{R}$, the coding rate of $\bar{R}$ equals to the flow rate of $\bar{R}$.
This conjecture is known as \emph{Network Coding Conjecture} (NCC) and we state a weaker version of it below.

For a directed graph $G = (V,E)$ we denote by $\un(G)$ the undirected graph $(V, \bar{E})$ obtained from $G$ by making each directed edge in $E$ undirected (i.e., replacing each $(u,v) \in E(G)$ by $\{u,v\}$). 
For a directed acyclic network $R = \bigl(G,c,(s_i,t_i)_{i \in [k]}\bigr)$ we define the undirected network $\un(R) = \bigl(\un(G), \bar{c}, (s_i,t_i)_{i \in [k]}\bigr)$ by keeping the source-target pairs and capacities the same, i.e, $c\bigl((u,v)\bigr) = \bar{c}\bigl(\{u,v\}\bigr)$.

\begin{conjecture}[Weaker NCC]
\label{conj:NCC}
 Let $R$ be a directed acyclic network, $r$ be a coding rate of $R$ and $\bar{r}$ be a flow rate of $\un(R)$.
 Then, $r = \bar{r}$.
\end{conjecture}

This conjecture was used to prove a conditional lower bound for sorting algorithms with an external memory~\cite{Farhadi19} and for circuits multiplying two numbers~\cite{Afshani19}.

\section{NCC Implies Data Structure Lower Bounds}
\label{sec:results}

In this paper, we provide several connections between lower bounds for data structures and other computational models.
The first connection is that NCC (Conjecture~\ref{conj:NCC}) implies lower bounds for data structures for the permutation inversion and the polynomial evaluation and interpolation.
Assuming NCC, we show that a query time $\Tm$ of a non-adaptive systematic data structure for any of the above problems satisfies $\Tm \geq \Omega\bigl(\log n / \log \log n\bigr)$, even if it uses linear space, i.e., the advice string $\Adv$ has size $\varepsilon n \log n$ for sufficiently small constant $\varepsilon > 0$.
Formally, we define $\Tm_\PI(\Sp)$ as a query time of the optimal non-adaptive systematic data structure for the permutation inversion using space at most $\Sp$.
Similarly, we define $\Tm^\F_\PE(\Sp)$ and $\Tm^\F_\PInt(\Sp)$ for the polynomial evaluation and interpolation over $\F$.

\begin{restatable}{theorem}{NccPerm}
 \label{thm:NCC_PI}
 Let $\varepsilon > 0$ be a sufficiently small constant.
 Assuming NCC, it holds that \item $\Tm_\PI(\varepsilon n \log n) \geq \Omega\bigl(\log n/\log \log n\bigr)$.
\end{restatable}

\begin{restatable}{theorem}{NccPoly}
 \label{thm:NCC_PE}
 Let $\F$ be a field and $n$ be a divisor of $|\F| - 1$.
 Let $\Sp = \varepsilon n \log |\F|$ for a sufficiently small constant $\varepsilon > 0$.
 Then assuming NCC, it holds that $\Tm^\F_\PE(\Sp), \Tm^\F_\PInt(\Sp) \geq \Omega\bigl(\log n/\log \log n\bigr)$.
\end{restatable}

Note that by Theorem~\ref{thm:NCC_PI}, assuming NCC, it holds that $\Sp \cdot \Tm  \geq \Omega\bigl(n \log^2 n/\log \log n\bigr)$ for $\Sp = \varepsilon n \log n$ and $\Tm = \Tm_\PI(\Sp)$.
The same holds for $\Tm^\F_\PE$ and $\Tm^\F_\PInt$ by Theorem~\ref{thm:NCC_PE}.
Thus, these conditional lower bounds cross the barrier $\Omega(n \log n)$ for $\Sp\cdot \Tm$ given by the best unconditional lower bounds known for the function inversion~\cite{Yao90,De10,Abusalah17,Gennaro05,Unruh07,Dodis17,Coretti18} and the lower bound for the succinct data structures for the polynomial evaluation by G\'{a}l and Miltersen~\cite{Gal03}.
The lower bound by Larsen~\cite{Larsen12} says that any cell probe data structure for the polynomial evaluation using linear space $\Sp_\cp = O(n \log n)$ needs at least logarithmic query time $\Tm_\cp \geq \Omega(\log n)$ if the size of the field is of super-linear size in $n$, i.e., $|\F| \geq n^{1 + \Omega(1)}$.
Then $\Sp_\cp \cdot \Tm_\cp \geq \Omega(n \log^2 n)$.
The lower bound given by Theorem~\ref{thm:NCC_PE} says that assuming NCC a non-adaptive data structure needs to read at least logarithmically many coefficients $\alpha_i$ of $p$ even if we know $\varepsilon n \log |\F|$ bits of information about the polynomial $p$ for free.
Our lower bound holds also for linear-size fields.

To prove Theorems~\ref{thm:NCC_PI} and~\ref{thm:NCC_PE}, we use the technique of Farhadi et al.~\cite{Farhadi19}.
The proof can be divided into two steps:
\begin{enumerate}
 \item From a data structure for the problem we derive a network $R$ with $O(\Tm n)$ edges such that $R$ admits an encoding scheme that is correct on a large fraction of the inputs. 
 This step is distinct for each problem and the reductions are shown in Sections~\ref{sec:NCC_function_inversion} and~\ref{sec:NCC_polynomial}.
 This step uses new ideas and interestingly, it uses the data structure twice in a sequence.
 \item If there is a network $R$ with $d n$ edges that admits an encoding scheme which is correct for a large fraction of inputs, then
$d \geq \Omega\bigl(\log n / \log \log n\bigr).$
 This step is common to all the problems.
 It was implicitly proved by Farhadi et al.~\cite{Farhadi19} and Afshani et al.~\cite{Afshani19}.
 For the sake of completeness, we give a proof of this step in Appendix~\ref{app:CorrectionGame}.
\end{enumerate}

\section{NCC Implies a Weak Lower Bound for the Function Inversion}
\label{sec:NCC_function_inversion}

In this section, we prove Theorem~\ref{thm:NCC_PI} that assuming NCC, any non-adaptive systematic data structure for the permutation inversion requires query time at least $\Omega\bigl(\log n / \log \log n\bigr)$ even if it uses linear space.
Let $\DS$ be a data structure for inverting permutations of a linear space $\Sp = \varepsilon n \log n$, for sufficiently small constant $\varepsilon < 1$, with query time $\Tm = \Tm_\PI(\Sp)$. 
Recall that $\Tm_\PI(\Sp)$ is a query time of the optimal non-adaptive systematic data structure for the permutation inversion using space $\Sp$.
From $\DS$ we construct a directed acyclic network $R = \bigl(G,c,(s_i,t_i)_{i \in [n]}\bigr)$ and an encoding scheme of a coding rate $\log n$.
By Conjecture~\ref{conj:NCC} we get that the flow rate of $\un(R) = \bigl(\bar{G},c,(s_i,t_i)_{i \in [n]}\bigr)$ is $\log n$ as well.
We prove that there are many source-target pairs of distance at least $\Omega(\log_\Tm n)$.
Since the number of edges of $\bar{G}$ will be $O(\Tm n)$ and flow rate of $\un(R)$ is $\log n$, we are able to derive a lower bound $\Tm \geq \Omega\bigl(\log n/\log \log n\bigr)$.

We construct the network $R$ in two steps.
First, we construct a network $R'$ that admits an encoding scheme $E'$ such that $E'$ is correct only on a substantial fraction of all possible inputs.
This might create correlations among messages received by the sources.
However, to use the Network Coding Conjecture we need to have a coding scheme that is able to reconstruct messages sampled from independent distributions.
To overcome this issue we use a technique introduced by Farhadi et al.~\cite{Farhadi19} and from $R'$ we construct a network $R$ that admits a correct encoding scheme.

Let $R = \bigl(G, c, (s_i,t_i)_{i \in [k]}\bigr)$ be a directed acyclic network.
Let each source receive a binary string of length $r$ as its input message, i.e., each $W_i = \{0,1\}^r$.
If we concatenate all input messages $w_i$ we get a string of length $r\cdot k$, thus the set of all possible inputs for an encoding scheme for $R$ corresponds to the set $\cI = \{0,1\}^{rk}$.
We say an encoding scheme is correct on an input $\bar{w} = (w_0,\dots,w_{k-1}) \in \cI$ if it is possible to reconstruct all messages $w_i$ at appropriate targets.
An $(\varepsilon,r)$-\emph{encoding scheme} is an encoding scheme which is correct on at least $2^{(1-\varepsilon)rk}$ inputs in $\cI$.

We say a directed network $R = \bigl(G,c,(s_i,t_i)_{i \in [k]}\bigr)$ is $(\delta, d)$-\emph{long} if for at least $\delta k$ source-target pairs $(s_i, t_i)$, it holds that distance between $s_i$ and $t_i$ in $\un(G)$ is at least $d$. Here, we measure the distance in the undirected graph $\un(G)$, even though the network $R$ is directed.
The following lemma is implicitly used by Farhadi et al.~\cite{Farhadi19} and Afshani et al.~\cite{Afshani19}.
We give its proof in Appendix~\ref{app:CorrectionGame} for the sake of completeness.
\begin{restatable}[Implicitly used in~\cite{Farhadi19, Afshani19}]{lemma}{SufficientCoding}
 \label{lem:SufficientCoding}
 Let $R = \bigl(G, q, (s_i,t_i)_{i \in [k]}\bigr)$ be a $(\delta, d)$-long directed acyclic uniform network for $\delta > \frac{5}{6}$ and sufficiently large $q \in \R^+$.
 Assume there is an $(\varepsilon,r)$-encoding scheme for $R$ for sufficiently small $\varepsilon$.
 Then assuming NCC, it holds that $\frac{|E(G)|}{k} \geq \delta'\cdot d$, where $\delta' = \frac{\delta - 5/6}{10}$.
\end{restatable}

Now we are ready to prove a conditional lower bound for the permutation inversion.
For the proof we use the following fact which follows from well-known Stirling's formula:

\begin{fact}
\label{fact:Perm}
 The number of permutations $[n] \to [n]$ is at least $2^{n\log n - 2n}$.
\end{fact}

\NccPerm*

\begin{proof}
 Let $\DS = \DS_n$ be the optimal data structure for the inversion of permutation on $[n]$ using space $\varepsilon n \log n$.
 We set $\Tm = \Tm_\PI(\varepsilon n \log n)$. 
 We will construct a directed acyclic uniform network $R = \bigl(G,r,(s_i,t_i)_{i \in [n]}\bigr)$ where $r = \log n$.
 Let $\varepsilon' = 2\cdot\varepsilon + \frac{2}{q} + \frac{2}{\log n}$ for sufficiently large $q$ so that we could apply Lemma~\ref{lem:SufficientCoding}.
 The network $R$ will admit an $(\varepsilon', r)$-encoding scheme $E$.
 The number of edges of $G$ will be at most $2\Tm n$ and the network $R$ will be $\bigl(\frac{9}{10},d\bigr)$-long for $d = \frac{1}{2} \log_{qt} n$.
 Thus, by Lemma~\ref{lem:SufficientCoding} we get that
 \[
  2\Tm = \frac{2 \Tm n}{n} \geq \Omega\left( \log_{qt} n \right),
 \]
 from which we can conclude that $\Tm \geq \Omega\bigl( \log n / \log \log n\bigr)$.
 Thus, it remains to construct the network $R$ and the scheme $E$.
 
 First, we construct a graph $G'$ which will yield the graph $G$ by deleting some edges.
 The graph $G'$ has three layers of $n$ vertices: a source layer $A$ of $n$ sources $s_0,\dots,s_{n-1}$, a middle layer $M$ of $n$ vertices $v_0,\dots,v_{n-1}$ and a target layer $B$ of $n$ vertices $u_0,\dots,u_{n-1}$.
 The targets $t_{0},\dots,t_{n-1}$ of $R$ will be assigned to the vertices $u_0,\dots,u_{n-1}$ later.
 
 We add edges according to the data structure $\DS$:
 Let $Q_j \subseteq [n]$ be a set of oracle queries, which $\DS$ makes during the computation of $f^{-1}(j)$, i.e., for each $i \in Q_j$, it queries the oracle of $f$ for $f(i)$.
 As $\DS$ is non-adaptive, the sets $Q_j$ are well-defined.
 For each $j \in [n]$ and $i \in Q_j$  we add edges $(s_i, v_j)$ and $(v_i, u_j)$. 
 We set a capacity of all edges to $r = \log n$.
 This finishes the construction of $G'$, see Fig.~\ref{fig:InvertingNetwork} for illustration of the graph $G'$.
 
  \begin{figure}
  \centering
  \includegraphics{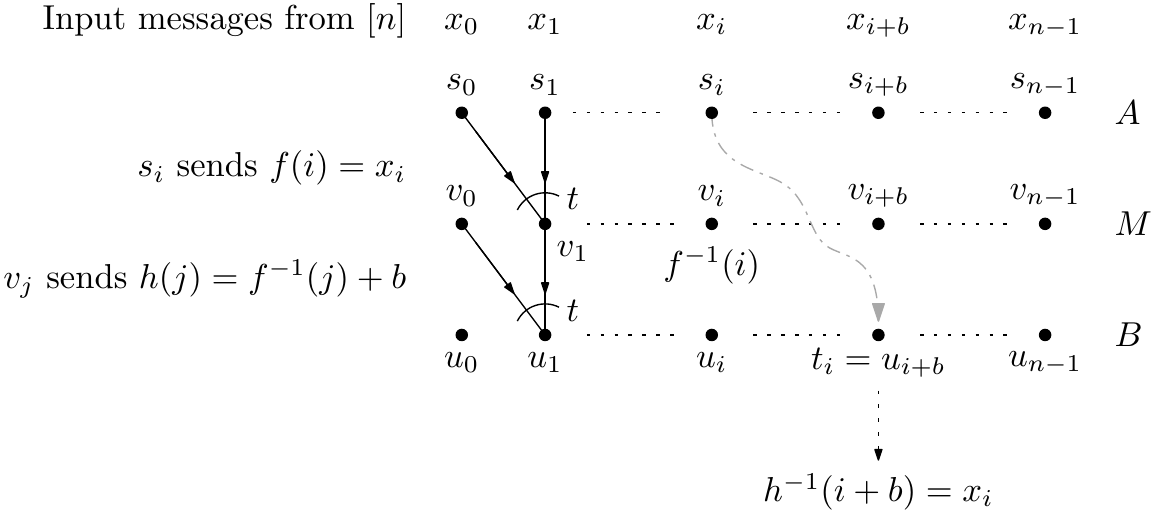}
  \caption{A sketch of the graph $G'$ and encoding scheme $E$.}
  \label{fig:InvertingNetwork}
 \end{figure}
 
 The graph $G'$ has exactly $2\Tm n$ edges.
 Moreover, the vertices of the middle and the target layer have in-degree at most $\Tm$ as the incoming edges correspond to the oracle queries made by $\DS$.
 However, some vertices of the source and the middle layer might have large outdegree, which is a problem that might prevent the network $R$ to be $\bigl(\frac{9}{10}, d\bigr)$-long.
 For example, the data structure $\DS$ could always query $f(0)$.
 Then, there would be edges $(s_0,v_j)$ and $(v_0,u_j)$ for all $j \in [n]$, hence all vertices would be at distance at most 4 in $\un(G')$.
 So we need to remove edges adjacent to high-degree vertices. Let $W \subseteq V(G')$ be the set of vertices of out-degree larger than $q\Tm$. 
 We remove all edges incident to $W$ from $G'$ to obtain the graph $G$. (For simplicity, we keep the degree 0 vertices in $G$).
 Thus, the maximum degree of $G$ is at most $q\Tm$.
 Since the graph $G'$ has $2\Tm n$ edges, it holds that $|W| \leq \frac{2}{q}\cdot n$.
 
 Now, we assign the targets of $R$ in such a way that $R$ is $\bigl(\frac{9}{10}, d\bigr)$-long.
 Let $C_v$ be the set of vertices of $G$ which have distance at most $d$ from $v$ in $\un(G)$.
 Since the maximum degree of $G$ is at most $q\Tm$ and $d = \frac{1}{2}\log_{q\Tm} n$, for each $v \in V(G)$,  $|C_v| \leq 2\sqrt{n}$.
 In particular, for every source $s_i$ it holds that $|C_{s_i} \cap B| \leq 2\sqrt{n}$, i.e., there are at most $2\sqrt{n}$ vertices in the target layer $B$ at distance smaller than $d$ from $v$.
 It follows from an averaging argument that there is an integer $b$ such that there are at least $n - 2\sqrt{n}$ sources $s_i$ with  distance at least $d$ from $u_{i + b}$ in $\un(G)$. (Here the addition $i + b$ is modulo $n$.)
 We fix one such $b$ and set $t_i = u_{i + b}$.
 For $n$ large enough, it holds that $n - 2\sqrt{n} \geq \frac{9}{10}\cdot n$.
 Thus, the network $R$ is $\bigl(\frac{9}{10}, d\bigr)$-long.
 
 It remains to construct the $(\varepsilon',r)$-encoding scheme $E$ for $R$ (see Fig.~\ref{fig:InvertingNetwork} for a sketch of the encoding $E$).
 Each source $s_i$ receives a number $x_i \in [n]$ as an input message.
 We interpret the string of the input messages $x_0, \dots, x_{n-1}$ as a function.
 We define the function $f: [n] \to [n]$ as $f(i) = x_i$.
 We will consider only those inputs $x_0, \dots, x_{n-1}$ which are pairwise distinct so that $f$ is a permutation.
 
 At a vertex $v_j$ of the middle layer $M$ we want to compute $f^{-1}(j)$ using the data structure $\DS$.
 To compute $f^{-1}(j)$ we need the appropriate advice string $\Advf{f}$ and answers to the oracle queries $Q_j$. 
 We fix an advice string $\Advf{f}$ to some particular value which will be determined later, and we focus only on inputs $x_0,\dots,x_{n-1}$ which  have the same advice string $\Advf{f}$.
 In $G'$ the vertex $v_j$ is connected exactly to the sources $s_i$ for $i \in Q_j$, but some of those connections might be missing in $G$.
 Thus for each $i$ such that $s_i \in W$, $x_i$ will be fixed to some particular value $c_i$ which will also be determined later.
 Each source $s_i$ sends the input $x_i$ along all outgoing edges incident to $s_i$.
 Thus, at a vertex $v_j$ we know the answers to all $f$-oracle queries in $Q_j$.
 Recall that $f(i) = x_i$ and each $x_i$ for $i \in Q_j$ was either fixed to $c_i$ or sent along the incoming edge $(s_i,v_j) \in E(G)$.
 We also know the advice string $\Advf{f}$ as it was fixed.
 Therefore, we can compute $f^{-1}(j)$ at every vertex $v_j$.
 Note that $f^{-1}(j)$ is the index of the source which received $j$ as an input message, i.e., if $f^{-1}(j) = i$, then $x_i = j$.
 
 Now, we define another permutation $h: [n] \to [n]$ as $h(j) = f^{-1}(j) + b$ where the addition is modulo $n$.
 Since $b$ is fixed, we can compute $h(j)$ at each vertex $v_j$.
 The goal is to compute $h^{-1}(\ell)$ at each vertex $u_\ell$ of the target layer.
 First, we argue that $h^{-1}(i + b) = x_{i}$.
 The permutation $f^{-1}$ maps an input message $x_i$ to the index $i$.
 The permutation $h$ maps an input message $x_i$ to the index $i + b$.
 Thus, the inverse permutation $h^{-1}$ maps the index $i + b$ to the input message $x_i$.
 If we are able to reconstruct $h^{-1}(i + b)$ at the target $t_i = u_{i + b}$, then in fact we are able to reconstruct $x_i$, the input message received by the source $s_i$.
 
 To reconstruct $h^{-1}(\ell)$ at the vertex $u_\ell$ we use the same strategy as for reconstructing $f^{-1}(j)$ at vertices $v_j$.
 We use again $\DS$, but this time for the function $h$.
 Again, we fix the advice string $\Advf{h}$ of $\DS$, and we fix $h(j)$ to some $d_j$ for each vertex $v_j \in W$.
 Each vertex $v_j$ sends the value $h(j)$ along all edges outgoing from $v_j$.
 To compute $h^{-1}(\ell)$ we need values $h(j)$ for all $j \in Q_\ell$, which are known to the vertex $u_\ell$. Again, they are either sent along the incoming edges or are fixed to $d_j$.
 Since the value of the advice string $\Advf{h}$ is fixed, we can compute the value $h^{-1}(\ell) = x_{\ell - b}$ at the vertex $u_\ell$.
 
 The network $R$ is correct on all inputs $x_0,\dots,x_{n-1}$ which encode a permutation and which are consistent with the fixed advice strings and the fixed values to the degree zero vertices.
 Now, we argue that we can fix all the values so that there will be many inputs consistent with them.
 By Fact~\ref{fact:Perm}, there are at least $2^{(n\log n) - 2n}$ inputs $x_0,\dots,x_{n-1}$ which encode a permutation.
 In order to make $R$ work, we fixed the following values:
 \begin{enumerate}
  \item Advice strings $\Advf{f}$ and $\Advf{h}$, in total $2\varepsilon \cdot n \log n$ bits.
  \item An input message $c_i$ for each source $s_i$ in $W$ and a value $d_j$ for each vertex $v_j$ in $W$.
  Since $|W| \leq \frac{2}{q}\cdot n$ and $c_i, d_j \in [n]$, we fix $\frac{2}{q} \cdot n \log n$ bits in total.
 \end{enumerate}
 Overall, we fix at most $(2\varepsilon + \frac{2}{q})\cdot n \log n$ bits.
 Thus, the fixed values divide the input strings into at most $2^{(2\varepsilon + \frac{2}{q})\cdot  n \log n}$ buckets. In each bucket all the input strings are consistent with the fixed values.
 We conclude that there is a choice of values to fix so that its corresponding bucket contains at least $2^{(1 - 2\varepsilon - \frac{2}{q} - \frac{2}{\log n} )\cdot n \log n} = 2^{(1 - \varepsilon')\cdot n \log n}$ input strings which encode a permutation. We pick that bucket and fix the corresponding values.
 Thus, the scheme $E$ is $(\varepsilon',r)$-encoding scheme, which concludes the proof.
 \end{proof}

\section{NCC Implies a Weak Lower Bound for the Polynomial Evaluation and Interpolation}
\label{sec:NCC_polynomial}
In this section, we prove Theorem~\ref{thm:NCC_PE}.
The proof follows the blueprint of the proof of Theorem~\ref{thm:NCC_PI}.
The construction of a network $R$ from a data structure is basically the same.
Thus, we mainly describe only an $(\varepsilon',r)$-encoding scheme for $R$.

\NccPoly*

\begin{proof}
	Let $\DS = \DS_n$ be the optimal non-adaptive systematic data structure for the evaluation of polynomials of degree up to $n - 1$ over $\F$ and using space $\Sp = \varepsilon n \log |\F|$.
	We set $t = t^\F_\PE(\Sp), r = \log |\F|$ and $\varepsilon' = 2\varepsilon + \frac{2}{q}$ for sufficently large $q$.
	Again, we will construct a network $R = \bigl(G,r,(s_i,t_i)_{i \in [n]}\bigr)$ from $\DS$.
	To construct an $(\varepsilon', r)$-encoding scheme for $R$, we use entries of FFFT, i.e., we will evaluate polynomials of degree at most $n-1$ in powers of a primitive $n$-th root of the unity.
	Thus, we fix a primitive $n$-th root of unity $\sigma \in \F$, which we know exists, as discussed in Section~\ref{sec:Polynomials}.
	
	We create a network $R$ from $\DS$ in the same way as we created in the proof of Theorem~\ref{thm:NCC_PI}.
	By Lemma~\ref{lem:SufficientCoding} we are able to conclude that $t \geq \Omega\bigl(\log n/\log \log n\bigr)$.
	First, we create a graph $G'$ of three layers $A = \{s_0,\dots,s_{n-1}\}, M = \{v_0,\dots,v_{n-1}\}$ and $B = \{u_0,\dots,u_{n-1}\}$ and we add $2tn$ edges to $G'$ according to the queries of $\DS$ -- on the vertex $v_j$ we will evaluate a polynomial in a point $\sigma^j$ and on the vertex $u_j$ we will evaluate a polynomial in a point $\sigma^{-j}$.
	Then, we create a graph $G$ from $G'$ by removing edges incident to vertices in a set $W$, which contains vertices of degree higher than $qt$.
	Finally, we set a shift $b \in [n]$ and set $t_i = u_{i + b}$ in such a way that the network $R$ is $\bigl(\frac{9}{10}, d\bigr)$-long for $d = \frac{1}{2}\log_{qt} n$.
	
	Now, we desribe an $(\varepsilon', r)$-encoding scheme $E$ for $R$ using $\DS$.
	Each source $s_i$ receives an input message $\alpha_i \in \F$ which we interpret as coefficients of a polynomial $p \in \F[x]$ (that is $p(x) = \sum_{i \in [n]} \alpha_i x^i$).
	Each source $s_i$ sends its input message $\alpha_i$ along all outgoing edges from $s_i$.
	Each vertex $v_j$ computes $p(\sigma^j)$ using $\DS$.
	Again, we fix the advice string $\Advf{p}$ and the input messages $\alpha_i$ for the sources $s_i$ in $W$.
	Each vertex $v_j$ computes a value $h(j) = p(\sigma^j)\cdot \sigma^{jb}$ and sends it along all outgoing edges from $v_j$.
	We define a new polynomial $p'(x) = \sum_{j \in [n]}h(j)x^j$.
	We fix the advice string $\Advf{{p'}}$ and the values $h(j)$ for each vertex $v_j \in W$.
	Thus, each vertex $u_\ell$ can compute a value $p'(\sigma^{-\ell})$.
	We claim that $\frac{p'(\sigma^{-\ell})}{n} = \alpha_{\ell - b}$.
    \begin{align*}
        \frac{p'(\sigma^{-\ell})}{n}
        &= \frac{1}{n} \sum_{j \in [n]} h(j) \sigma^{-\ell j} = \frac{1}{n}  \sum_{j \in [n]} p(\sigma^j) \sigma^{jb} \sigma^{-\ell j} 
        = \frac{1}{n} \sum_{j \in [n]} \left( \sum_{i \in [n]} \alpha_{i} \sigma^{ji} \right) \sigma^{jb} \sigma^{-\ell j} \\
        &= \frac{1}{n} \sum_{i \in [n]} \alpha_{i} \left( \sum_{j \in [n]} \sigma^{ji} \sigma^{jb} \sigma^{-\ell j}  \right) 
        = \frac{1}{n} \sum_{i \in [n]} \alpha_{i} \left( \sum_{j \in [n]} \sigma^{j(i + b - \ell)} \right) \\
        &= \alpha_{(\ell - b \mod n)}
    \end{align*}
    The last equality is by noting that $\sum_{j \in [n]} \sigma^{j(i + b - \ell)}=n$ for $i=\ell-b$ and $0$ otherwise.
	Therefore, at each target $t_{i} = u_{i + b}$ we can reconstruct the input message $\alpha_i$.
	See Fig.~\ref{fig:PolyScheme} for a sketch of the scheme $E$.
	
	\begin{figure}
	 \centering
	 \includegraphics{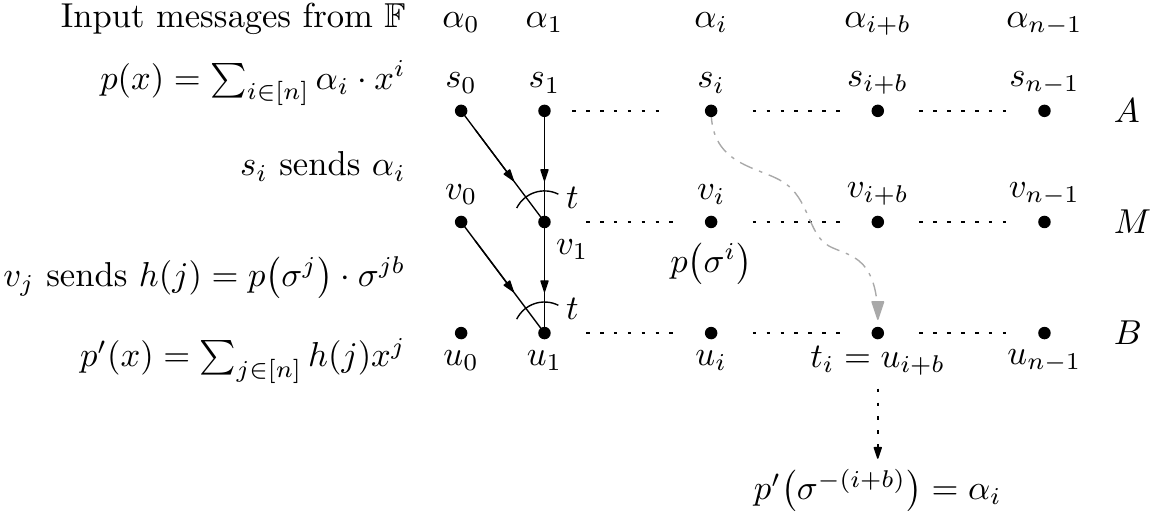}
	 \caption{Sketch of the encoding scheme $E$.}
	 \label{fig:PolyScheme}
	\end{figure}

	Again, we can fix values of advice strings $\Advf{p}$ and $\Advf{{p'}}$ (at most $2\varepsilon\cdot n\log|\F|$ fixed bits), input messages $\alpha_i$ for each $s_i \in W$ and value of $h(j)$ for each $v_j \in W$ (at most $\frac{2}{q}\cdot n \log |\F|$ fixed bits) in such a way there is a set $\cF$ of inputs $(\alpha_0,\dots,\alpha_{n-1})$ consistent with such fixing and $|\cF| \geq 2^{(1 - 2\varepsilon - \frac{2}{q})n \log |\F|}$.
	Therefore, the scheme $E$ is $(\varepsilon',r)$-encoding scheme.
	This finishes the proof that $\Tm^\F_\PE(\Sp) \geq \Omega\bigl(\log n / \log \log n\bigr)$.
	
	Essentially, the same proof can be used to prove the lower bound for $\Tm^\F_\PInt(\Sp)$.
	Note that, the data structure $\DS$ is used only for evaluating some polynomials in powers of the primitive root $\sigma$, i.e., computing entries of $\text{FFFT}_{n,\sigma}(\alpha_0,\dots,\alpha_{n-1})$.
	However as discussed in Section~\ref{sec:Polynomials}, it holds that $\text{FFFT}_{n,\sigma} = n \cdot \text{FFFT}^{-1}_{n,\sigma^{-1}}$.
	Moreover, entries of $\text{FFFT}^{-1}_{n,\sigma^{-1}}(\beta_0,\dots,\beta_{n-1})$ can be computed by a data structure for the polynomial interpolation.
	Thus, we may replace both uses of the data structure for the polynomial evaluation with a data structure for the polynomial interpolation.
	Therefore, we can use a data structure for the interpolation as $\DS$ and with slight changes of $R$ and $E$, we would get again an $(\varepsilon',r)$-encoding scheme.
\end{proof}

\section{Strong Lower Bounds for Data Structures and Lower Bounds for Boolean Circuits}
\label{sec:valiant}

In this section, we study a connection between non-adaptive data structures and boolean circuits.
We are interested in circuits with binary AND and OR gates, and unary NOT gates. (See e.g.~\cite{jukna2012boolean} for background on circuits).

Corrigan-Gibbs and Kogan~\cite{Gibbs19} describe a connection between lower bounds for non-adaptive data structures and lower bounds for boolean circuits for a special case when the data structure computes function inversion.
They show that we would get a circuit lower bound if any non-adaptive data structure using $O(n^{\varepsilon})$ queries must use at least $\omega\bigl(n \log n / \log \log n\bigr)$ bits of advice (for some fixed constant $\varepsilon > 0$).
To be able to formally restate their theorem we present some of their definitions.
We define a \emph{boolean operator} to be a family of functions $\left( F_n \right)_{n \in \mathbb{N}}$ for $F_n \colon \left\{ 0,1 \right\}^{n} \rightarrow \left\{ 0,1 \right\}^{n}$ represented by boolean circuits with $n$ input and $n$ output bits and constant fan-in gates.
A boolean operator is said to be an \emph{explicit operator} if the decision problem whether the $j$-th output bit of $F_n$ is equal to one is in the complexity class NP.

\begin{theorem}[Corrigan-Gibbs and Kogan~\cite{Gibbs19}, Theorem~3 (verbatim)]
	If every explicit operator has fan-in-two boolean circuits of size $O(n)$ and depth $O(\log n)$ then, for every $\varepsilon > 0$, there exists a family of strongly non-adaptive black-box algorithms that inverts all functions $f \colon [N] \rightarrow [N]$ using $O\bigl(N \log N / \log \log N\bigr)$ bits of advice and $O(N^\varepsilon)$ online queries.
	\label{thm:a:corrigangibbs_kogan_thm3_verbatim}
\end{theorem}

To prove their theorem Corrigan-Gibbs and Kogan~\cite{Gibbs19} use the common bits model of boolean circuits described by Valiant~\cite{valiant1977graph,valiant1983exponential,valiant1992boolean}.
Valiant proves that for any circuit there is a small cut, called \emph{common bits}, such that each output bit is connected just to few input bits (formally stated in Theorem~\ref{thm:a:valiant_ckt}).
Corrigan-Gibbs and Kogan~\cite{Gibbs19} use the common bits of the given circuit to create a non-adaptive data structure by setting the advice string to the content of common bits and the queries are to those function values which are still connected to the particular output after removing the common bits.

Observe that it follows from the proof of Theorem~\ref{thm:a:corrigangibbs_kogan_thm3_verbatim} that the hard explicit operator is turning the function table into the table of its inverse function.
The theorem is therefore slightly stronger in the sense that if we have a data structure lower bound we also have a lower bound for a concrete boolean operator.
It is also not straightforward to state a connection between circuits computing FFFT and non-adaptive data structures computing polynomial evaluation (resp. polynomial interpolation) as a consequence of Theorem~\ref{thm:a:corrigangibbs_kogan_thm3_verbatim}.
Thus, we restate the Valiant's result to be able to state a more general theorem.

\begin{theorem}[Valiant~\cite{valiant1977graph,valiant1983exponential,valiant1992boolean}]
	For every constant $\varepsilon > 0$, for every family of constant fan-in boolean circuits $\left\{ C_n \right\}_{n \in \mathbb{N}}$, where $C_n \colon \left\{ 0,1 \right\}^n \rightarrow \left\{ 0,1 \right\}^{n}$ is of size $O(n)$ and depth $O(\log n)$, and for every $n \in \mathbb{N}$ it holds that the circuit $C_n$ contains a set of gates called \emph{common bits} of size $O\bigl(n / \log \log n\bigr)$ such that if we remove those gates then each output bit is connected to at most $O(n^\varepsilon)$ input bits.
	\label{thm:a:valiant_ckt}
\end{theorem}

Now we can state a general theorem translating lower bounds for non-adaptive data structures to circuit lower bounds.
This allows us to apply the theorem directly to many different problems.

\begin{theorem}
	For every $n \in \mathbb{N}$ let us define $b(n) = \lceil \log n \rceil$ and for every function $f_n \colon \left\{ 0,1 \right\}^{n \cdot b(n)} \rightarrow \left\{ 0,1 \right\}^{n \cdot b(n)}$ and for every $i \in [n]$ we define a function $f_{n,i}\colon \left\{ 0,1 \right\}^{n \cdot b(n)} \rightarrow \left\{ 0,1 \right\}^{b(n)}$ as follows: 
	$f_{n,i}(x) = f_n(x)_{i \cdot  b(n), i \cdot b(n) + 1, \ldots, (i+1) \cdot b(n) - 1}, $
	i.e., $f_{n,i}$ returns the $(i+1)$-st consecutive block of $b(n)$ bits of the output of $f_{n}$.

	If there is a size $O(n\log n)$ and depth $O(\log n)$ circuit family $\{C_n\}_{n \in \mathbb{N}}$,
	where $C_n$ evaluates a function $f_n \colon \left\{ 0,1 \right\}^{n \cdot b(n)}\rightarrow \left\{ 0,1 \right\}^{n \cdot b(n)}$,
	then for every constant $\varepsilon > 0$ there exists a family of non-adaptive data structures $\left\{ \DS_n \right\}_{n \in \mathbb{N}}$,
	where $\DS_n$ on input $x \in \left\{ 0,1 \right\}^{n \cdot b(n)}$  uses $O\bigl(n \log n/ \log \log n\bigr)$ bits of advice and on a query $j \in [n]$ answers $f_{n,j}(x)$ using $O(n^{\varepsilon})$ queries to the input. 
	\label{thm:a:general_thm_from_valiant}
\end{theorem}

The proof of Theorem~\ref{thm:a:general_thm_from_valiant} is the same as the proof of Theorem~3 of Corrigan-Gibbs and Kogan~\cite{Gibbs19} (restated here as Theorem~\ref{thm:a:corrigangibbs_kogan_thm3_verbatim}).
Note that the data structures are not uniform in the sense that the algorithms for producing the advice string and for answering queries may differ for different input sizes.
If we would like to get a uniform algorithm we would need the assumption that the explicit operator has linear size and logarithmic depth {\em uniform} circuits.

Let us state concrete instances of Theorem~\ref{thm:a:general_thm_from_valiant}.
First, we formally state the stronger version of Theorem~3 of Corrigan-Gibbs and Kogan~\cite{Gibbs19}, which follows from their proof.
\begin{corollary}
	If there is a circuit family $\left\{ C_n \right\}_{n \in \mathbb{N}}$, such that $C_n \colon \left\{ 0,1 \right\}^{n \lceil \log n \rceil} \rightarrow \left\{ 0,1 \right\}^{n \lceil \log n \rceil}$ is of size $O(n\log n)$ and depth $O(\log n)$ and inverts a function $f_n \colon [n] \rightarrow [n]$ (given on input as a function table) on all points (i.e., returns function table\footnote{When $f_n$ is not a permutation we allow a table of any function which has zero if there is no preimage and any preimage if there are more possibilities.} of $f_n^{-1}$), then for every constant $\varepsilon > 0$ there exists a family of non-adaptive data structures $\left\{ \DS_n \right\}_{n \in \mathbb{N}}$ such that $\DS_n$ on all input functions $f_n \colon [n] \rightarrow [n]$ uses $O\bigl(n \log n/ \log \log n\bigr)$ bits of advice and for any $x \in [n]$ it answers $f_n^{-1}(x)$ using $O(n^{\varepsilon})$ queries to the input. 
	\label{thm:a:corrigangibbs_kogan_restated}
\end{corollary}

Theorem~\ref{thm:a:general_thm_from_valiant} is general enough to easily capture the problem of computing FFFT over a finite field.
Note that by the connection of FFFT to polynomial evaluation and interpolation the following corollary captures both problems.

\begin{corollary}
	Let $\mathcal{S} = \left\{ p^k \mid p \text{ is a prime}, k \in \mathbb{N}, k \neq 0 \right\}$ be the set of all sizes of finite fields.
	For each $n \in \mathcal{S}$, let $\F_n = GF(n)$ and $\sigma_n$ be a primitive $(n-1)$-th root of unity (thus a generator of the multiplicative group $\F_n^*$).

	If there is a circuit family computing $\text{FFFT}_{n-1, \sigma_n}$ (over $\F_n$) of size $O(n \log n)$ and depth $O(\log n)$ (where each input and output number is represented by $\log |\F_n|$ bits) then for every $\varepsilon > 0$ there is a family of non-adaptive data structures $\left\{\DS_n\right\}_{n \in \mathcal{S}}$ where $\DS_n$ uses advice of size $O\bigl(n \log n / \log \log n\bigr)$ and on a query $j \in [n-1]$ outputs the $j$-th output of $\text{FFFT}_{n-1, \sigma_n}$ using $O(n^\varepsilon)$ queries to the input.
	\label{cor:a:valiant_evaluate}
\end{corollary}


To put the corollary in counter-positive way: if for some $\varepsilon>0$, there are no non-adaptive data structures for polynomial interpolation, polynomial evaluation or FFFT with advise of size $o\bigl(n \log n / \log \log \log n\bigr)$ that use $O(n^\varepsilon)$ queries to the input then there are no linear-size circuits of logarithmic depth for FFFT.




In Theorem~\ref{thm:NCC_PI}, resp. Theorem~\ref{thm:NCC_PE}, we prove a conditional lower bound for permutation inversion, resp. polynomial evaluation and polynomial interpolation, of the form, that a non-adaptive data structure using $\varepsilon n \log n$ bits must do at least $\Omega\bigl(\log n/\log \log n\bigr)$ queries.
It is not clear if assuming NCC we can get a sufficiently strong lower bound which would rule out non-adaptive data structures with sublinear advice string using $O(n^{\varepsilon})$ oracle queries.

\begin{corollary}
	We say that a circuit $C_n \colon \left\{ 0,1 \right\}^{n \lceil \log n \rceil} \rightarrow \left\{ 0,1 \right\}^{n \lceil \log n \rceil}$ sorts its input if on an input viewed as $n$ binary strings $x_1, x_2, \ldots, x_n \in \left\{ 0,1 \right\}^{\lceil \log n \rceil}$ outputs the strings sorted lexicographically.

	If there is a circuit family $\left( C_n \right)_{n \in \mathbb{N}}$, where $C_n \colon \left\{ 0,1 \right\}^{n \lceil \log n \rceil} \rightarrow \left\{ 0,1 \right\}^{n \lceil \log n \rceil}$ sorts its inputs, and each circuit $C_n$ is of size $O(n \log n)$ and depth $O(\log n)$ then for every $\varepsilon > 0$, for every permutation $f:[n] \rightarrow [n]$ there is a non-adaptive data structure for inverting $f$ that uses advice of size $O\bigl(n \log n / \log \log n\bigr)$ and $O(n^\varepsilon)$ queries.
	\label{cor:a:valiant_sort}
\end{corollary}


The works of Farhadi~et~al.~\cite{Farhadi19} and Asharov~et~al.~\cite{asharov2021sorting} connect the NCC conjecture directly to lower bounds for sorting.
Their work studies sorting $n$ numbers of $k+w$ bits by their first $k$ bits.
Namely Asharov~et~al.~\cite{asharov2021sorting} show that NCC implies that constant fan-in constant fan-out circuits must have size $\Omega\bigl(nk(w - \log(n) + k)\bigr)$ whenever $w > \log(n) - k$ and $k \leq \log n$. This is incomparable to our results as we have $w=0$.


\bibliography{main_all_in_one_file}

\appendix

\section{Proof of Lemma~\ref{lem:SufficientCoding}}
\label{app:CorrectionGame}
\SufficientCoding*

For the proof, we use an $\cF$-correction game introduced by Farhadi et al.~\cite{Farhadi19} (the statement of the following definition and lemma is due to Afshani et al.~\cite{Afshani19}).
\begin{definition}[$\cF$-correction game~\cite{Farhadi19, Afshani19}]
 Let $\cF \subseteq \{0,1\}^{m\ell}$. 
 The $\cF$-\emph{correction game} with $\ell + 1$ players is defined as follows. 
The game is played by $\ell$ ordinary players $p_0, \dots, p_{\ell - 1}$ and one designated supervisor player $\bar{u}$. 
The supervisor $\bar{u}$ receives $\ell$ strings $\alpha_0, \dots, \alpha_{\ell - 1} \in \{0,1\}^{m}$ chosen independently at random. 
For every $i \in [\ell]$, $\bar{u}$ sends $p_i$ a message $\beta_i$. 
Given $\beta_i$, the player $p_i$ produces a string $\gamma_i \in \{0,1\}^{m}$ such that $\bigl((\alpha_0 \oplus \gamma_0),\dots,(\alpha_{\ell-1} \oplus \gamma_{\ell-1})\bigr) \in \cF$.
\end{definition}

\begin{lemma}[\cite{Farhadi19, Afshani19}]
\label{lem:CorrectionGame}
 If $|\cF| \geq 2^{(1 - \varepsilon)m\ell}$, then there exists a protocol for the $\cF$-correction game with $\ell + 1$ players such that the messages $(\beta_i)_{i \in [\ell]}$ are prefix-free and
 \begin{equation}
  \sum_{i \in [\ell]} \E |\beta_i| \leq 3 \ell + 2 \ell \log \left( \sqrt{\frac{\varepsilon}{2}} \cdot m + 1 \right) + \sqrt{\frac{\varepsilon}{8}}\cdot m\ell \log \frac{2}{\varepsilon}.
  \label{eq:AdviceLength}
 \end{equation}
\end{lemma}

Observe that for sufficiently small $\varepsilon$ and sufficiently large $m$ the formula in Equation~\ref{eq:AdviceLength} can be bounded by $\frac{m\ell}{4}$.
Thus, we can suppose that the expected total length of the messages sent by supervisor $u$ in the $\cF$-correction game is at most $\frac{m\ell}{4}$.

\begin{proof}[Proof of Lemma~\ref{lem:SufficientCoding}]
 Let $R = \bigr(G, q, (s_i, t_i)_{i \in [k]}\bigl)$ be a network given by the assumption of the lemma.
 We will create a directed acyclic network $R' = \bigr(G', c', (s'_i,t_i)_{i \in [k]}\bigr)$ which will admits a correct encoding scheme.
 Thus, we will be able to apply NCC to $R'$.
 Note that, the network $R'$ has new sources $s'_i$ but the original targets $t_i$.
 
 The network $R'$ is defined as follows.
 We add new sources $s'_1,\dots,s'_k$ and one special vertex $u$ to the graph $G$.
 For each $i$ we add the following new edges:
 \begin{itemize}
  \item Edge $(s'_i, s_i)$ and $(s'_i, u)$ of capacity $r$, i.e, edges connecting the new sources with the original ones and the new special vertex.
  \item Edges connecting the new special vertex $u$ with the original sources $s_i$ and the targets $t_i$, i.e, the edges $(u, s_i)$ and $(u, t_i)$ of capacity $\E|\beta_i|$, where $\beta_i$ is the message sent by the supervisor $\bar{u}$ to the player $p_i$ in the protocol for the $\cF$-correction game given by Lemma~\ref{lem:CorrectionGame}.
 \end{itemize}

This finishes the construction of $R'$.
By assumption, there is a set $\cF \subseteq \{0,1\}^{kr}$ and an encoding scheme $E$ for $R$ such that $E$ is correct on inputs in $\cF$.
Note that $R$ is a subnetwork of $R'$.
Thus, to create an encoding scheme for $R'$ which will be correct on every input in $\{0,1\}^{kr}$ we use an encoding $E$ to recover some messages $w = (w_0,\dots,w_{k-1}) \in \cF$ and the special vertex $u$ which will send messages as the supervisor $\bar{u}$ in the $\cF$-correction game.
After that, the targets $t_i$ will be able to reconstruct the input messages $w'_i$ received at the new sources $s'_i$.

More formally, let $w'_i$ be an input message received at the source $s'_i$. 
Each $w'_i$ is uniformly sampled from $\{0,1\}^r$ (independently on other $w'_j$).
Now, the encoding scheme $E'$ for $R'$ works as follows:
\begin{enumerate}
 \item Each source $s'_i$ sends the input message $w'_i$ to the vertex $s_i$ and $u$.
 \item The vertex $u$ computes the messages $\beta_0,\dots,\beta_{k-1}$ according to the protocol given by Lemma~\ref{lem:CorrectionGame} (applied for the messages $\alpha_i = w'_i$).
 Then for each $i \in [k]$, the vertex $u$ sends the messages $\beta_i$ to the vertex $s_i$ and $t_i$.
 \item Each vertex $s_i$ computes the string $\gamma_i$.
 By Lemma~\ref{lem:CorrectionGame}, it holds that 
 \[
\bigl((w'_0 \oplus \gamma_0), \dots,(w'_{k-1} \oplus \gamma_{k-1})\bigr) \in \cF.  
 \]

 Thus, we can use the encoding scheme $E$ for $R$ to reconstruct strings $w'_i \oplus \gamma_i$ at each target $t_i$.
 \item Each target $t_i$ can reconstruct strings $w'_i \oplus \gamma_i$ and $\gamma_i$. Thus, it can reconstruct the input message $w'_i$.
\end{enumerate}
By construction of the network $R'$, it is clear that the encoding scheme $E'$ respects the capacities $c'$.

The encoding scheme $E'$ witnesses that the coding rate of $R'$ is at least $r$.
Thus by NCC (Conjecture~\ref{conj:NCC}), we conclude that the flow rate of $\un(R')$ is at least $r$ as well, i.e., there is a multicommodity flow $F = (f^0,\dots,f^{k-1})$ for $\un(R')$ which transports at least $r$ units of each commodity $i$.
Now, we argue that there is only a small fraction of the total flow which goes through the special vertex $u$.

\begin{claim}
\label{clm:SmallUFlow}
 The size of total flow which goes through $u$ is at most $\frac{3}{4} kr$.
\end{claim}
\begin{claimproof}
The total capacity of the edges incident to the vertex $u$ is at most $\frac{3}{2} kr$.
 The vertex $u$ is incident to $k$ edges $\{u,s'_i\}$ of capacity $r$, which contribute by $kr$ to the total capacity.
 Then for each $i$, the vertex $u$ is incident to the edges $\{u,s_i\}$ and $\{u,t_i\}$, which have both capacity $\E|\beta_i|$.
 By Lemma~\ref{lem:CorrectionGame}, we have that $\sum_i \E |\beta_i| \leq \frac{1}{4} kr$.
 Thus, these edges contribute by $\frac{1}{2} kr$ to the total capacity.
  By conservation of the flow, it must hold that
\begin{align*}
 \sum_{i \in [k]} \sum_{v \in V(G')} f^i(v, u) &\leq \frac{3}{4} kr, \\
 \sum_{i \in [k]} \sum_{v \in V(G')} f^i(u, v) &\leq \frac{3}{4} kr.
\end{align*} 
\end{claimproof}

Let $A \subseteq [k]$ be a set of indices of source-target pairs $(s_i, t_i)$ such that at least $\frac{r}{10}$ units of the commodity $i$ do not go through the vertex $u$.
It follows that the set $A$ is substantially large.

\begin{claim}
\label{clm:GoodSet}
$|A| \geq \frac{k}{6}$. 
\end{claim}
\begin{claimproof}
 Suppose opposite, $|A| < \frac{k}{6}$, i.e., there are at least $\frac{5k}{6}$ source-target pairs $(s_i, t_i)$ such that strictly more than $\frac{9r}{10}$ units of the commodity $i$ goes through the vertex $u$.
 Therefore, the total size of the flow going through $u$ is strictly larger than $\frac{5k}{6} \cdot \frac{9r}{10} = \frac{3}{4} kr$, which contradicts Claim~\ref{clm:SmallUFlow}.
\end{claimproof}

Let $L \subseteq [k]$ be a set of indices of pairs $(s_i, t_i)$ such that their distance in $\un(G)$ is at least $d$.
By the assumption of the lemma, it holds that $|L| \geq \delta k$.
Note that for each $i$, the distance between $s_i$ and $t_i$ in $G'$ is 2 because of the vertex $u$.
However, due to Claim~\ref{clm:GoodSet} there is a lot of source-target pairs $(s_i,t_i)$ which are far in $\un(G)$ and some units of the commodity $i$ do not go through $u$:
\[
 \bigl|A \cap L\bigr| \geq \bigl|A\bigr| - \bigl|[k] \setminus L\bigr| = \left(\delta - \frac{5}{6} \right)\cdot k.
\]

Let $L' = A \cap L$, i.e., the set $L'$ contains indices $i \in [k]$ such that distance between $s_i$ and $t_i$ in $\un(G)$ is at least $d$ and at least $\frac{r}{10}$ units of the commodity $i$ do not go through $u$ -- thus, it has to go through paths of length at least $d$.
Now, we are ready to prove the assertion of the lemma.
Let $\bar{E} = E\bigl(\un(G)\bigr)$.
\begin{align*}
 r \cdot |\bar{E}| = \sum_{e \in \bar{E}} c(e) & \geq  \sum_{\{v,w\} \in \bar{E}} \sum_{i \in [k]} f^i(v,w) + f^i(w,v) \\
 & \geq \sum_{i \in L'} \sum_{\{v,w\} \in \bar{E}} f^i(v,w) + f^i(w,v) \\
 & \geq \left(\delta - \frac{5}{6} \right)k \cdot \frac{r}{10}d = \delta'\cdot kr\cdot d & \text{By definition of $L'$.}
\end{align*}
It follows that $\frac{|\bar{E}|}{k} \geq \delta'\cdot d$.
\end{proof}

\end{document}